\documentclass[conference]{IEEEtran}

\usepackage{amsmath,amsfonts,bm}

\def\eqref#1{(\ref{#1})}

\def\1{\bm{1}}

\newcommand{\bx}{{\mathbf x}}

\newcommand{\bu}{{\mathbf u}}

\newcommand{\bz}{{\mathbf z}}

\newcommand{\bA}{{\mathbf A}}
\newcommand{\bD}{{\mathbf D}}

\newcommand{\bL}{{\mathbf L}}

\newcommand{\bLambda}{{\boldsymbol \Lambda}}

\newcommand{\mE}{{\mathbb E}}

\newcommand{\bB}{{\mathbf B}}

\newcommand{\cE}{{\mathcal E}}

\newcommand{\cG}{{\mathcal G}}

\newcommand{\cV}{{\mathcal V}}

\def\bU{{\bf U}}

\def\minwrt[#1]{\underset{#1}{\text{minimize}}}
\def\maxwrt[#1]{\underset{#1}{\text{max}}}

\ifCLASSINFOpdf
\else
   \usepackage[dvips]{graphicx}
\fi
\usepackage{url}

\usepackage{amsmath,amssymb,euscript,yfonts,psfrag,latexsym,graphicx}
\usepackage{bbm,color,amstext,wasysym,parskip}
\graphicspath{{./},{../figures/},{./figures/}}
\usepackage{caption}
\usepackage{subcaption}
\usepackage{float}
\usepackage{array}
\usepackage{url}
\usepackage{algorithm}
\usepackage{algorithmic}
\usepackage{xcolor}

\usepackage{amsthm}

\usepackage{graphicx}

\newtheorem{prop}{Proposition}

\newtheorem{property}{Property}

\usepackage{soul}

\begin{document}



\title{Graph Bispectrum for Nonlinear Mode Interactions
}

\author{\IEEEauthorblockN{Rahul Singh, Reza Abiri, Walter Besio}
\IEEEauthorblockA{\textit{Department of Electrical, Computer, and Biomedical Engineering} \\
\textit{University of Rhode Island}\\
Kingston, USA \\}
}

\maketitle

\begin{abstract}
We introduce a graph bispectrum formulation for characterizing higher-order interactions in graph signals. While conventional graph spectral methods capture only second-order structure, many graph signals exhibit nonlinear interactions that are not reflected in covariance or graph power spectra. Motivated by classical higher-order spectral analysis, we define a graph bispectrum tensor based on third-order moments of graph Fourier coefficients and derive a compact graph bicoherence measure that summarizes nonlinear mode interactions in a low-dimensional and scale-invariant form. We establish key properties of the proposed quantities, including vanishing third-order moments for Gaussian graph signals and a dynamical interpretation in terms of nonlinear mode coupling. Experiments on synthetic random graph signals demonstrate that the proposed measures detect complementary nonlinear dependencies even when second-order statistics are similar. We further apply the method to EEG recordings from the CHB-MIT Scalp EEG Database and show that ictal activity exhibits substantially increased nonlinear graph spectral coupling compared to interictal periods. The proposed approach provides an interpretable and computationally efficient tool for higher-order interaction analysis for graph signals.
\end{abstract}


\section{Introduction}

\IEEEPARstart{H}{igher}-order spectral methods such as the bispectrum and bicoherence have been extensively used to detect nonlinear interactions and phase coupling between frequency components~\cite{Bri91,SwaGiaZho97,ColWhiHam98}. The bispectrum captures third-order correlations in the Fourier domain and provides information about frequency coupling beyond conventional power spectra. These methods have been successfully applied in areas including radar, communications, geophysics, and biomedical signal analysis.

Graph signal processing (GSP)~\cite{CheShiWri20,ManChaSin18,SinChaMan16} provides a framework for analyzing signals defined on irregular domains such as social networks, sensor arrays, and brain connectivity graphs. By leveraging the eigendecomposition of the graph Laplacian, the graph Fourier transform (GFT) enables a notion of frequency that generalizes classical spectral analysis to graph-structured data. However, most existing graph spectral methods rely primarily on second-order statistics, such as covariance and graph power spectra, which describe only pairwise correlations and energy distributions. While effective for linear systems, these quantities do not capture nonlinear interactions between graph modes that arise in many complex dynamical systems. Such higher-order interactions are especially relevant in applications such as neuroscience, where large-scale brain activity often exhibits nonlinear synchronization and cross-scale coupling.

Extending bispectral analysis to graph signals is nontrivial. Graph frequencies, given by Laplacian eigenvalues, are generally irregular and do not admit a natural frequency addition. Consequently, classical bispectrum formulations cannot be directly applied to graph domains. This creates a gap in GSP: while graph spectral methods provide a notion of frequency, there is no established framework for analyzing higher-order interactions between graph modes. In this paper, we address this by introducing a graph bispectrum formulation for signals defined on graphs. We first define a graph bispectrum tensor that captures third-order interactions between graph Fourier coefficients. To obtain a practical and interpretable representation, we propose a compact graph bicoherence, which aggregates nonlinear interactions onto each graph mode through projection of a quadratic signal and normalization. This yields a low-dimensional, scale-invariant measure of nonlinear interaction strength across graph frequencies. We further define a normalized graph bispectral energy as a global measure of higher-order interaction.

We validate the proposed formulation on both synthetic and real data. On synthetic signals generated over random graphs, we show that the proposed measures can detect nonlinear interactions even when graph spectra are similar. We then apply the method to scalp electroencephalogram (EEG) recordings from the CHB-MIT Scalp EEG Database~\cite{MIT_CHB}, demonstrating that ictal (seizure) activity exhibits significantly increased nonlinear coupling compared to interictal periods. These results indicate that higher-order graph spectral analysis provides information complementary to conventional graph spectral methods.

\textbf{Related Work:} Higher-order spectral analysis, particularly the bispectrum, has long been used to characterize nonlinear interactions and phase coupling in time-series signals. Extensions of bispectral analysis to non-Euclidean domains have been studied in group-theoretic settings, including bispectral invariants and skew-spectrum constructions for functions on groups and graphs~\cite{Kon08,KonBor08,Kak09,Kak12,MatMatSan24}. In parallel, graph signal processing has established a spectral framework based on Laplacian eigenvectors, enabling graph Fourier analysis, spectral filtering, and graph power spectra. However, existing graph spectral methods largely rely on second-order statistics and do not explicitly characterize nonlinear interactions between graph modes. The proposed approach defines a graph bispectrum directly in the graph Fourier domain through a harmonic interaction algebra induced by eigenvector multiplication, providing an interpretable and computationally tractable characterization of higher-order nonlinear mode coupling in graph signals.

\section{Background}
\label{sec:background}

\subsection{Classical Bispectrum}
\label{subsec:classical_bispectrum}
For a zero-mean, stationary time series $x(t)$ with Fourier transform $X(f)$, the bispectrum is defined as~\cite{SwaGiaZho97}
\begin{equation}
    \label{eq:bispect}
    B(f_1, f_2) = \mathbb{E}\big[X(f_1)\, X(f_2)\, X^*(f_1 + f_2)\big],
\end{equation}
where $(\cdot)^*$ denotes complex conjugation. The bispectrum measures third-order correlations in the frequency domain and captures phase coupling, i.e., interactions in which two frequency components combine to produce a third at their sum frequency.

A key property of the bispectrum is that it vanishes for linear Gaussian processes, making it a sensitive indicator of nonlinearity and departures from Gaussianity. Unlike the power spectrum, which depends only on magnitude, the bispectrum encodes phase relationships and thus provides insight into the structure of nonlinear interactions.

To obtain a scale-invariant measure, the bicoherence is defined as a normalized version of the bispectrum:
\begin{equation}
b(f_1, f_2) =
\frac{|B(f_1, f_2)|}
{\sqrt{\mathbb{E}[|X(f_1)X(f_2)|^2]\;\mathbb{E}[|X(f_1+f_2)|^2]}}.
\end{equation}
The bicoherence takes values in [0,1] and quantifies the strength of phase coupling independent of signal amplitude.

These definitions rely on the additive structure of the Fourier domain, where sums of frequencies are well defined. This enables nonlinear interactions to be interpreted as frequency coupling. In graph spectral domains, however, Laplacian eigenvalues do not generally exhibit a corresponding addition rule. This motivates a graph-specific formulation of bispectral analysis.

\subsection{Graph Frequency Analysis}
GSP is concerned with the generalization of classical signal processing concepts and tools to graph signals. GSP relates the vertex and spectral domains of a graph, much as classical signal processing connects the time and frequency domains of a time series~\cite{CheShiWri20,ManChaSin18}.
The eigenvalues and eigenvectors of the graph Laplacian provide a notion of frequency for signals defined on a graph. The graph Laplacian eigenvectors associated with low frequencies, vary slowly across the graph, i.e., if two vertices are connected by an edge, the values of the eigenvector at those locations are likely to be similar. The eigenvectors associated with larger eigenvalues oscillate more rapidly and are more likely to have dissimilar values on vertices connected by an edge. The graph Fourier transform and its inverse  give us a way to equivalently represent a signal in two different domains: the vertex domain and the graph spectral domain. 

Let $\cG = (\cV,\cE)$ be a graph, where $\cV$ is the set of $N$ number of nodes and $\cE$ is the set of edges. The adjacency matrix of the graph is denoted as $\bA\in \mathbb{R}^{N\times N}$ and has entries from $\{0,1\}$. A nonzero entry in $\bA$ indicates the presence of an edge between two nodes, i.e., $A_{ij} = 1$, if nodes $i$ and $j$ are connected, and $A_{ij} = 0$ otherwise. The graph Fourier (spectral) analysis relies on the spectral decomposition of graph Laplacians. The traditional combinatorial graph Laplacian is defined as $\bL = \bD - \bA$, with $\bD = \mathrm{diag}\{ d_1,d_2,\ldots, d_N \}$ and $d_i  = \sum_j A_{ij}$. Based on the eigendecomposition of the graph Laplacian $\bL = \bU \bLambda \bU^T$, where $\bU \in \mathbb{R}^{N\times N}$ comprises of orthonormal eigenvectors and $\bLambda = \mathrm{diag}\{\lambda_1, \ldots, \lambda_N\}$ is a diagonal matrix of eigenvalues, the graph Fourier transform is defined with eigenvectors of the graph Laplacian being the graph Fourier modes (harmonics) and the corresponding eigenvalues being the graph frequencies~\cite{ShuNarFro13}. Assuming $\lambda_1 \leq \lambda_2 \leq \ldots \leq \lambda_N$, $\lambda_1$ corresponds to the lowest (zero) frequency and $\lambda_N$ corresponds to the highest frequency of the graph. Let $\bx \in \mathbb{R}^{N}$ be a graph signal, then the graph Fourier transform (GFT) and the inverse Fourier transform are defined as $\hat{\bx} =  \bU^T \bx$ and $\bx = \bU \hat{\bx}$, respectively.  

\section{Main Results}
\label{sec:main}
Let $\bx \in \mathbb{R}^N$ be a random graph signal defined on an undirected graph with Laplacian eigendecomposition $\bL = \bU \Lambda \bU^\top$, where $\bU = [\bu_1,\ldots,\bu_N]$ contains orthonormal graph Fourier modes. The graph Fourier transform is $\hat{\bx} = \bU^\top \bx$, with coefficients $\hat{x}_k$. Conventional graph spectral analysis relies primarily on second-order statistics such as the graph power spectrum
$\mE \left[ |\hat{x}_k|^2 \right]$, which characterizes the distribution of signal energy across modes. More generally, second-order quantities of the form $\mE \left[  \hat{x}_i \hat{x}_j \right]$ 
capture pairwise correlations between graph modes. While useful for describing covariance and spectral energy distributions, such quantities do not characterize nonlinear interactions between graph modes. To address this limitation, we introduce higher-order graph spectral measures based on third-order moments.

We define the graph bispectrum tensor $\bB$ with its entries as 
\begin{equation}
B_{ijk} = \mathbb{E}[\hat{x}_i \hat{x}_j \hat{x}_k], \quad i,j,k = 1,2,\ldots,N,
\end{equation}
which captures third-order interactions between graph Fourier coefficients. Unlike classical Fourier analysis, graph frequencies do not admit a natural additive structure. Instead, interactions between graph modes are governed by the Laplacian eigenvectors themselves. In particular, interactions between graph modes are characterized through the pointwise product of graph harmonics
\begin{equation}
\bu_i \odot \bu_j = \sum_k C_{ijk} \bu_k,
\end{equation}
where
\begin{equation}
C_{ijk} = \sum_{n=1}^N u_i(n)u_j(n)u_k(n)
\end{equation}

are the mode interaction coefficients. The coefficient $C_{ijk}$ measures how strongly graph modes $i$ and $j$ combine to contribute to mode $k$, thereby replacing the role of frequency addition in the classical bispectrum. Its relation to classical bispectrum is given by the following proposition.

\begin{prop}
\label{prop:relation}
Consider a cycle graph with $N$ nodes, whose Laplacian eigenvectors coincide with the discrete Fourier basis. Then the graph harmonic interaction coefficients satisfy
\begin{equation}
C_{ijk} \neq 0
\quad \text{only if} \quad
k \equiv i+j \; (\mathrm{mod}\; N).
\end{equation}
Consequently, the proposed graph bispectrum recovers the classical frequency coupling structure underlying the conventional bispectrum.
\end{prop}

\begin{proof}
For a cycle graph, the graph Fourier modes are discrete complex exponentials of the form
\begin{equation*}
u_k(n)=\frac{1}{\sqrt{N}}e^{j2\pi kn/N}.
\end{equation*}

Substituting these eigenvectors into the interaction coefficient definition gives
\begin{align*}
C_{ijk} &=
\sum_{n=1}^{N}
u_i(n)u_j(n){u_k^*(n)}\\
&= \frac{1}{N}
\sum_{n=1}^{N}
e^{j2\pi(i+j-k)n/N}.
\end{align*}

This sum equals $1$ when $k = i+j \ (\mathrm{mod}\ N)$ and $0$ otherwise. Thus, graph harmonic interactions on a cycle graph obey the same frequency coupling relation as the classical Fourier basis, recovering the interaction structure underlying the conventional bispectrum.
\end{proof}

\begin{figure*}[h]
\vspace{-0.4cm}
\centering
\begin{subfigure}[t]{0.32\textwidth}
\centering	
\includegraphics[scale=0.33]{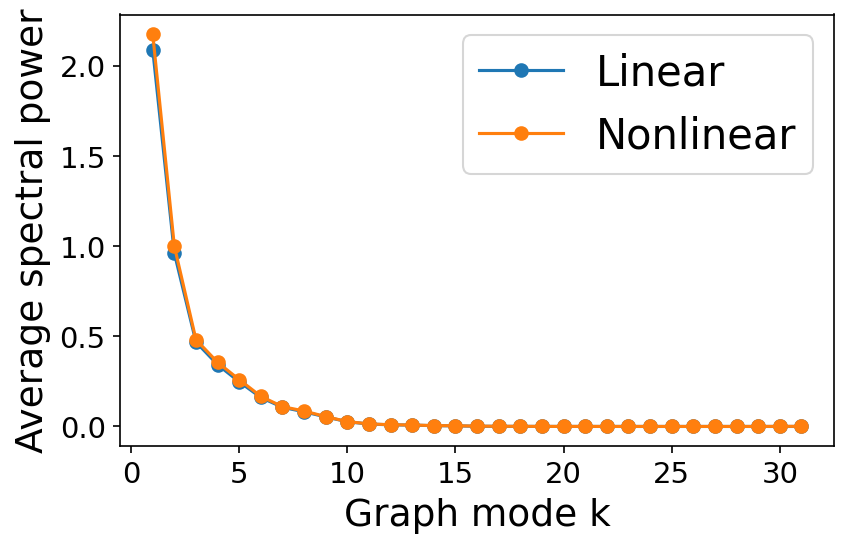}
\vspace{-0.2cm}
\caption{}
\label{fig:toy1a}
\end{subfigure}
\begin{subfigure}[t]{0.32\textwidth}
\centering	
\includegraphics[scale=0.33]{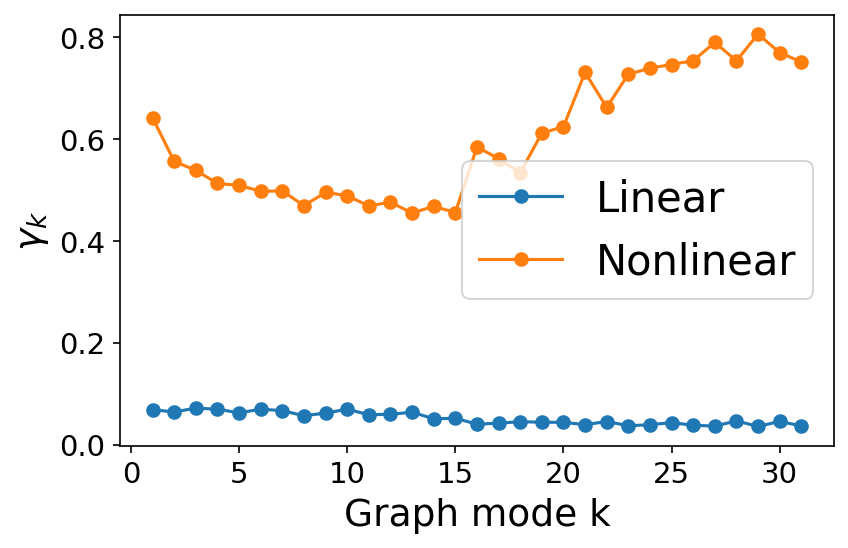}
\vspace{-0.2cm}
\caption{}
\label{fig:toy1b}
\end{subfigure}
\begin{subfigure}[t]{0.32\textwidth}
\centering
\includegraphics[scale=0.33]{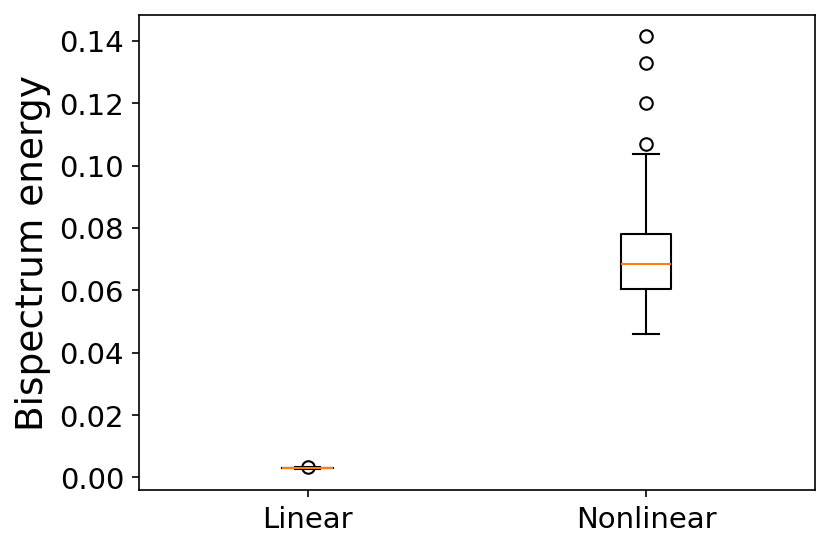}
\vspace{-0.2cm}
\caption{}
\label{fig:toy1c}
\end{subfigure}
\vspace{-0.2cm}
\caption{Synthetic data experiments: (a) Average spectral power, (b) graph bicoherence, and (c) normalized bispectral energy.}
\label{fig:toy1}
\end{figure*}
\subsection{Compact Graph Bicoherence}
The full graph bispectrum tensor $\bB$ is high-dimensional ($N^3$ entries) and challenging to estimate and interpret in practice. To obtain a compact and robust representation, we consider the quadratic signal $\bx \odot \bx$, which, when expanded in the graph Fourier basis, contains all pairwise interactions between graph modes. We then project this signal onto the graph Fourier basis to obtain a mode-wise representation of nonlinear interactions.
Based on this, we define a compact interaction statistic that aggregates all nonlinear contributions associated with each graph mode:
\begin{equation}
\label{eq:compact_statistic}
b_k = \mathbb{E}[\hat{x}_k \widehat{(x^2)}_k].
\end{equation}
We further define the normalized compact graph bicoherence
\begin{equation}
\gamma_k =
\frac{|b_k|}
{\sqrt{\mathbb{E}[|\hat{x}_k|^2]\;\mathbb{E}[|\widehat{(x^2)}_k|^2]}},
\end{equation}
which quantifies the normalized alignment between graph mode $k$ and its nonlinear interaction component. This construction gives a low-dimensional ($N$), scale-invariant representation of higher-order interactions that is both interpretable and computationally efficient.
Relationship between compact graph bicoherence with the full graph bispectrum is given by the following proposition.

\begin{prop}
The compact interaction statistic satisfies
\begin{equation}
\label{eq:compact_relation}
b_k = \sum_{i,j} C_{ijk} B_{ijk}.
\end{equation}
\end{prop}

\begin{proof}
Expanding $\bx = \sum_i \hat{x}_i \bu_i$, we obtain
\begin{equation*}
\bx^2 = \bx \odot \bx = \sum_{i,j} \hat{x}_i \hat{x}_j (\bu_i \odot \bu_j)
= \sum_{i,j,k} C_{ijk} \hat{x}_i \hat{x}_j \bu_k.
\end{equation*}
Taking the graph Fourier transform gives
\begin{equation}
\label{eq:squared_GFT}
\widehat{(x^2)}_k =  \sum_{i,j} C_{ijk} \hat{x}_i \hat{x}_j.
\end{equation}
Multiplying by $\hat{x}_k$ and taking expectations, we get \eqref{eq:compact_relation}.
\end{proof}


\subsection{Statistical Properties and Dynamical Interpretation}
\label{subsec:properties}
We next summarize basic statistical properties of the proposed quantities. These properties show that the graph bispectrum behaves analogously to the classical bispectrum while remaining well defined for graph signals.

\begin{property}
Let $\bx \in \mathbb{R}^N$ be a zero-mean Gaussian graph signal. Then
\begin{equation}
B_{ijk}=0, \quad \forall i,j,k, ~\mathrm{and~} b_k = 0, \quad \forall k.
\end{equation}
\end{property}

\begin{proof}
Since $\bx$ is Gaussian and the graph Fourier transform is linear, $\hat{\bx}=\bU^\top \bx$ is also a zero-mean Gaussian random vector. For any zero-mean Gaussian vector, all odd-order moments vanish. Hence
\begin{equation*}
\mathbb{E}[\hat{x}_i\hat{x}_j\hat{x}_k]=0,
\end{equation*}
which implies $B_{ijk}=0$. Using \eqref{eq:compact_relation}, we obtain $b_k=0$.
\end{proof}

\begin{property}
The normalized compact graph bicoherence $\gamma_k$ satisfies $0 \leq \gamma_k \leq 1$ and is invariant under scaling $\bx \mapsto \alpha \bx, ~ \forall \alpha \neq 0$.
\end{property}

\begin{property}
The normalized bispectral energy
\begin{equation}
E_{\mathrm{norm}} =
\frac{1}{N^3}
\sum_{i,j,k=1}^N
\frac{
\left(\mathbb{E}[\hat{x}_i\hat{x}_j\hat{x}_k]\right)^2
}
{
\mathbb{E}[\hat{x}_i^2]
\mathbb{E}[\hat{x}_j^2]
\mathbb{E}[\hat{x}_k^2]
}
\end{equation}
is nonnegative, scale-invariant, and vanishes for zero-mean Gaussian graph signals. This follows directly from the definition and properties of Gaussian moments.
\end{property}

Beyond their statistical interpretation, the proposed graph bispectral quantities also admit a natural dynamical interpretation in terms of nonlinear interactions between graph modes. Consider a nonlinear graph dynamical system of the form
\begin{equation}
\dot{\bx} = -\bL \bx + \beta (\bx \odot \bx) + \boldsymbol{\eta},
\end{equation}
where $-\bL \bx$ represents graph diffusion or linear propagation over the network, $\bx \odot \bx$ models nonlinear local interactions, $\beta$ controls the strength of nonlinear coupling,
and $\boldsymbol{\eta}$ denotes external input or stochastic noise. This can be viewed as a graph analogue of nonlinear reaction-diffusion systems, where activity evolves both through graph propagation and nonlinear self-interaction. Projecting the signal dynamics onto the graph Fourier basis gives
\begin{equation}
\dot{\hat{x}}_k = -\lambda_k \hat{x}_k + \beta \widehat{(x^2)}_k + \hat{\eta}_k,
\end{equation}
where $\widehat{(x^2)}_k = \bu_k^T (\bx \odot \bx)$ represents the nonlinear contribution to $k^{th}$ graph mode, as also given by \eqref{eq:squared_GFT}. This shows that the dynamics of graph mode $k$ depends not only on its linear diffusion term $-\lambda_k \hat{x}_k$, but also on nonlinear interactions between other graph modes.

Furthermore, the compact interaction statistic $b_k$ given by \eqref{eq:compact_statistic} measures the statistical coupling between graph mode $k$ and the nonlinear interaction term driving its dynamics. Large values of $b_k$ therefore indicate strong nonlinear mode coupling and nonlinear energy transfer into graph mode $k$.



\section{Experimental Results}
\label{sec:results}

\subsection{Synthetic Data}
\label{subsec:toy_example}

To validate the proposed graph bispectrum formulation, we first consider a synthetic random graph experiment comparing linear Gaussian signals with nonlinear signals generated via quadratic interactions. We create a $32$ node Erd\H{o}s-R\'enyi graph~\cite{ManChaSin18} with edge probability of 0.2. We generate low-pass filtered linear signals and non-linear signals as following
\begin{align*}
    \bx_{l} &= \bU h(\Lambda) \bz,\quad \bz \sim \mathcal{N}(\mathbf{0}, \mathbf{I}) \\
    \bx_{nl} &= \bx_{l} + \alpha (\bx_{l}  \odot \bx_{l} ).
\end{align*}
We choose $h(\lambda) = \exp{(- \tau \lambda)}$ with $\tau = 0.7$ and value of $\alpha$ is set to $0.7$. We used 400 realizations of these random graph signals for our analysis.

The graph spectrum (Figure~\ref{fig:toy1a}) shows that both linear and nonlinear signals exhibit nearly identical second-order statistics, confirming that energy distribution alone cannot distinguish between them. In contrast, the proposed compact graph bicoherence (Figure~\ref{fig:toy1b}) clearly separates the two cases, with significantly higher values for the nonlinear signal across all modes. This demonstrates that the proposed measure captures higher-order interactions that are not captured by covariance-based methods. The normalized nonlinear interaction energy (Figure~\ref{fig:toy1c}) further validates this effect, with a clear distinction between linear and nonlinear signals. These results validate that the graph bispectrum framework detects nonlinear dependencies even when second-order statistics are similar.
\begin{figure*}[h]
\vspace{-0.4cm}
\centering
\begin{subfigure}[t]{0.32\textwidth}
\centering	
\includegraphics[scale=0.33]{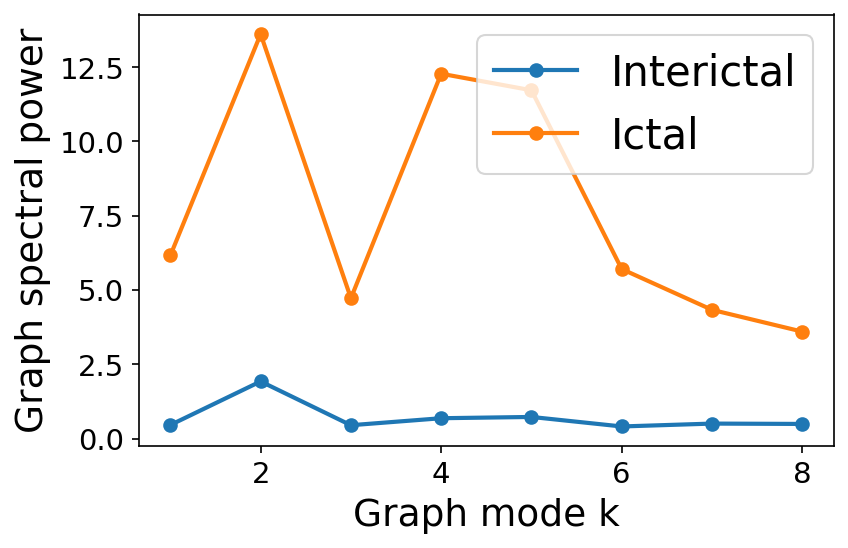}
\vspace{-0.2cm}
\caption{}
\label{fig:eeg1a}
\end{subfigure}
\begin{subfigure}[t]{0.32\textwidth}
\centering	
\includegraphics[scale=0.33]{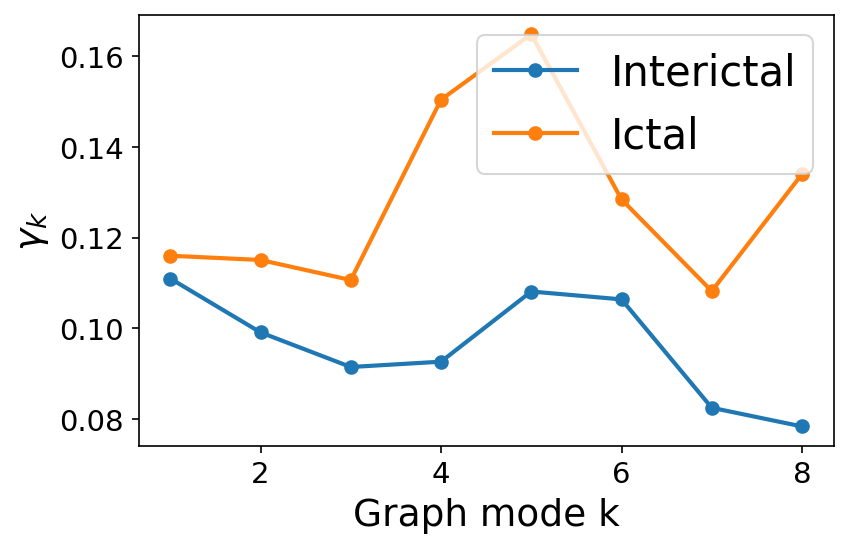}
\vspace{-0.2cm}
\caption{}
\label{fig:eeg1b}
\end{subfigure}
\begin{subfigure}[t]{0.32\textwidth}
\centering
\includegraphics[scale=0.33]{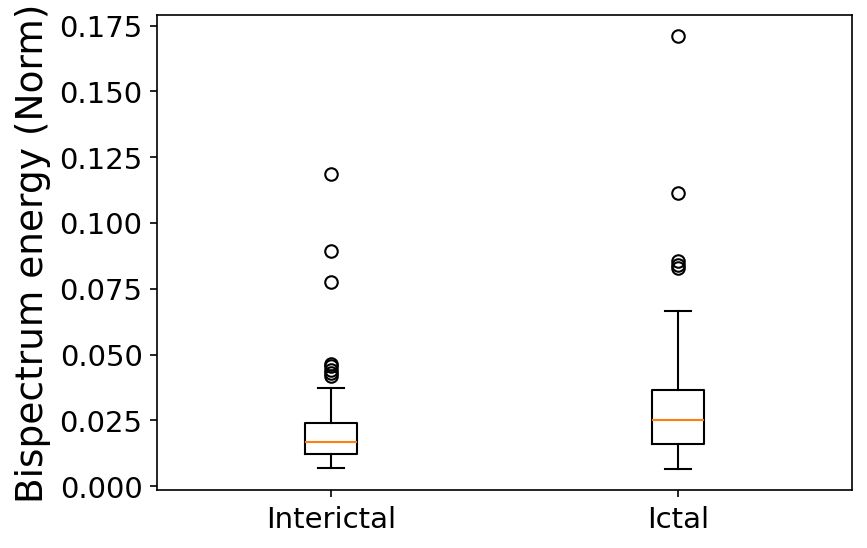}
\vspace{-0.2cm}
\caption{}
\label{fig:eeg1c}
\end{subfigure}
\vspace{-0.2cm}
\caption{EEG data experiments: (a) Average spectral power, (b) graph bicoherence, and (c) normalized bispectral energy.}
\label{fig:eeg1}
\end{figure*}
\subsection{Real-world Data}
\label{subsec:EEG}
We next evaluate the proposed graph bispectral measures on scalp EEG (electroencephalogram) recordings from the CHB-MIT Scalp EEG Database~\cite{MIT_CHB}, a widely used pediatric epilepsy dataset containing EEG recordings with annotated seizure intervals. The dataset consists of multi-channel scalp EEG recordings acquired at a sampling rate of 256 Hz using the international 10-20 electrode system. We analyzed subject chb01 (23 EEG channels), which contains multiple seizure recordings with annotated ictal intervals. For each seizure recording, EEG segments were partitioned into overlapping temporal windows of length 5 seconds with overlapping window of 2.5 seconds. Segments with annotated seizure intervals were labeled as ictal, while segments outside seizure intervals were labeled as interictal. To avoid contamination from seizure transitions and preictal activity, windows within 60 seconds of seizure onset or offset were excluded from the interictal set. A total of 80 ictal and 80 interictal segments were used for analysis. EEG recordings were bandpass filtered between 1-40 Hz and subsequently normalized channel-wise using z-score normalization. To represent EEG activity as a graph signal, we constructed a graph whose nodes correspond to EEG channels. We used pairwise Euclidean distances between channel locations and $k$-nearest-neighbor graph ($k=4$) was constructed. 

Figure~\ref{fig:eeg1a} shows the graph spectral power for ictal and interictal EEG windows across first $8$ graph modes. Ictal activity exhibits substantially increased graph spectral energy across nearly all low-frequency graph modes, reflecting stronger spatial synchronization and large-scale network activation during seizures. Figure~\ref{fig:eeg1b} shows the proposed compact graph bicoherence. Compared to interictal windows, ictal activity exhibits consistently elevated bicoherence across several graph modes. Because the compact graph bicoherence is normalized, the observed increases suggest enhanced nonlinear interactions between graph modes during seizure activity rather than merely reflecting overall signal amplitude differences. To quantify global higher-order interactions, we additionally computed the normalized graph bispectrum energy across graph modes. As shown in Figure~\ref{fig:eeg1c}, ictal windows exhibit substantially larger bispectral energy than interictal windows, even after normalization. This suggests that seizure dynamics are characterized not only by increased graph spectral power but also by stronger higher-order graph spectral coupling.

Overall, these results suggest that seizure dynamics are characterized not only by increased graph spectral power but also by enhanced nonlinear coupling between graph modes. The proposed graph bispectral quantities therefore provide complementary information beyond conventional graph spectral analysis and offer an interpretable representation of higher-order network interactions in EEG signals.
\section{Conclusion}
In this work, we introduced a graph bispectrum formulation for characterizing higher-order interactions in graph signals and proposed a compact graph bicoherence measure for capturing nonlinear graph mode coupling. We established key  properties and demonstrated through synthetic and EEG experiments that the proposed quantities detect nonlinear interactions beyond conventional graph spectral analysis. In particular, seizure EEG exhibited substantially increased graph bispectral coupling compared to interictal activity, suggesting enhanced higher-order coordination during seizures. Future work will explore graph neural network architectures that explicitly incorporate graph bispectral features, enabling learned models to capture not only pairwise message passing but also higher-order nonlinear interactions between graph modes.





\bibliography{references}{}
\bibliographystyle{IEEEtran}

\end{document}